\documentclass[12pt,reqno]{amsart}

\usepackage{epsfig}
\usepackage{amsmath, amsthm, amsfonts}
\usepackage{graphicx}

\numberwithin{equation}{section}

\newcommand{\isup}[3]{\underset{#1\leq #2\leq #3}{\textrm{sup }}}

\newcommand{\D}{{\mathbb D}}

\newcommand{\C}{{\mathbb C}}

\newcommand{\Z}{{\mathbb Z}}
\newcommand{\T}{{\mathbb T}}
\renewcommand{\d}{\partial}

\newcommand{\f}{\varphi}

\newtheorem{theo}{{\sc \bf Theorem}}[section]
\newtheorem{cor}[theo]{{\sc \bf Corollary}}
\newtheorem{lem}[theo]{{\sc \bf Lemma}}
\newtheorem{prop}[theo]{{\sc \bf Proposition}}

\begin{document}

\title{Global boundary conditions for a Dirac operator on the solid torus}

\author{Slawomir Klimek}
\address{Department of Mathematical Sciences,
Indiana University-Purdue University Indianapolis,
402 N. Blackford St., Indianapolis, IN 46202, U.S.A.}
\email{sklimek@math.iupui.edu}

\author{Matt McBride}
\address{Department of Mathematical Sciences,
Indiana University-Purdue University Indianapolis,
402 N. Blackford St., Indianapolis, IN 46202, U.S.A.}
\email{mmcbride@math.iupui.edu}

\thanks{}

\date{\today}

\begin{abstract}
We study a Dirac operator subject to Atiayh-Patodi-Singer like boundary conditions on the solid torus and show that the corresponding boundary value problem is elliptic, in the sense that the Dirac operator has a compact parametrix.
\end{abstract}

\maketitle
\section{Introduction}

The celebrated Atiyah-Patodi-Singer boundary condition \cite{APS} for Dirac operators on closed manifolds with boundary was introduced as a key ingredient in the generalization of the Atiyah-Singer index theorem.   It is a non-local boundary condition which makes the Dirac operator Fredholm.   An advantage of the APS condition  is that further detailed analysis can be carried out where local conditions such as Dirichlet and Neumann may not be well behaved.   

That theory works, however, under the assumption that both the manifold and the operator have a cylindrical structure near the boundary. A semi-infinite cylinder can then be smoothly attached to the manifold and the Dirac operator can be naturally extended over to the cylinder. The APS boundary condition can be then described in the following geometrical terms: a sufficiently regular section is in the domain of the operator if it extends to a square integrable solution on the cylinder. The issue though is that many concrete natural operators do not have a cylindrical structure near the boundary.
Even a simple d-bar operator  $\d/\d\overline{z}$ on a disk in the complex plane does not have this structure. The solid torus example studied here is not cylindrical near boundary either.

In this paper we consider a Dirac operator on the solid torus considered geometrically as the product of the unit disk and the unit circle. We construct another non-local boundary condition similar in spirit to the APS boundary condition. It was inspired by the non-local boundary conditions discussed in \cite{MS} and \cite{FSGS} to get around the necessity of having a cylindrical structure near the boundary.   The boundary condition we propose in this paper has the same geometrical interpretation as the APS condition.   Namely, we consider the solid torus as a subset of the bigger noncompact space of the plane cross the unit circle. We define the domain of the Dirac operator in full analogy with APS as consisting of those sufficiently regular sections which extend to  square integrable solutions on the complement of the solid torus.

The motivation for studying this particular example comes from our larger project of developing a concept of a noncommutative manifold with boundary and noncommutative elliptic boundary conditions. This is done by studying examples, starting with two-dimensional domains and continuing with more complex cases.
In particular our efforts were concentrated on studying quantum analogs of Dirac operators subject to APS like boundary conditions, see \cite{CKW}, \cite{KM1}, \cite{KM2}, and \cite{KM3}.  The solid torus studied in this paper is possibly the simplest three dimensional example, yet significantly more difficult then the two-dimensional examples studied in our previous papers. While the standard APS theory does not apply, the example however  seems to have an attractive noncommutative version. In the spirit of the series, this paper will be followed up by another study that discusses the same type of theory and calculations for the quantum analog to the solid torus.

We show in this paper that the Dirac operator on the solid torus subject to our non-local boundary condition is self-adjoint and has no kernel. Using a partial Fourier transform we obtain an explicit formula for its inverse. We then show that the inverse is a compact operator, which is the main feature of elliptic boundary conditions. Moreover we prove that the inverse is a $p$-th Schatten class operator for $p>3$. This is obtained by direct analysis of the formula for the inverse using subtle estimates involving modified Bessel functions.

The paper is organized as follows.   In section 2 the Hilbert space and the Dirac operator $D$  are defined, and the boundary condition is stated.   The section also contains the computation of the kernel of the Dirac operator, and the computation of its inverse $Q$.   The proof of the main theorem, the compactness of $Q$, and the Schatten class computation is contained in section 3.   In the appendix we collect numerous facts about the modified Bessel functions, some are classical and some are more recent.

\section{The Dirac Operator}

We  begin with the necessary notation.   Let $\D = \left\{z\in \C \ : \ |z|\leq 1\right\}$ be the unit disk and let $S^1 = \left\{e^{i\theta}\in \C \ : \ 0\leq\theta\leq 2\pi\right\}$ be the unit circle.
Also let $\T^2$ be the two dimensional torus and let $ST^2$ be the solid torus: $ST^2 = \D \times S^1\subset\C\times S^1$.  The boundary of $ST^2$ is $\T^2$.  The operators that we are studying will be acting in the Hilbert space $\mathcal{H} = L^2(ST^2,\C^2) \cong L^2(ST^2)\otimes \C^2$, i.e. the space of square-integrable complex vector-valued functions on the solid torus.  The inner product of $F,G\in\mathcal{H}$ will be denoted as $\langle F, G \rangle$.

We proceed to the definitions of the main object we study in this paper.
We consider the following formally self-adjoint Dirac operator $D$ defined on $\mathcal{H}$ by

\begin{equation}\label{dirac_operator}
D = \left(
\begin{array}{cc}
\frac{1}{i}\frac{\d}{\d\theta} & 2\frac{\d}{\d\overline{z}} \\
-2\frac{\d}{\d z} & -\frac{1}{i}\frac{\d}{\d\theta}
\end{array}\right).
\end{equation}
Notice that $D$ can act on either $L^2(\D\times S^1)\otimes \C^2$ or $L^2(\C \times S^1)\otimes \C^2$.   

The domain of $D$ is defined to be: 
\begin{equation}\label{dirac_domain}
\textrm{dom}(D) = \left\{F\in H^1(ST^2)\otimes \C^2\ : \ \exists F^{ext}\in H_{loc}^1\left((\C\times S^1)\setminus ST^2\right)\otimes \C^2\right\}
\end{equation}
such that (1), (2), and (3) hold:

\begin{enumerate}
\item $F^{ext}|_{\T^2} = F|_{\T^2}$,
\item $DF^{ext} = 0$,
\item $F^{ext}\in L^2\left((\C\times S^1)\setminus ST^2\right)\otimes \C^2$.
\end{enumerate}
Here $H^1$ is the first Sobolev space.

The first task is to study the kernel of $D$.  
For a function $F\in L^2(ST^2)\otimes \C^2$, by using the polar decomposition $z=re^{i\f}$, we have the Fourier decomposition as follows:

\begin{equation}\label{FourierDec}
F = \sum_{m,n\in\Z} \left(
\begin{array}{c}
f_{m,n}(r) \\
g_{m,n}(r)
\end{array}\right) e^{in\f + im\theta}.
\end{equation}
The norm of $F$ can then be expressed as:
\begin{equation*}
\|F\|^2=\langle F, F \rangle=\sum_{m,n\in\Z}\int_0^1 \left(|f_{m,n}|^2+|g_{m,n}|^2\right)r\,dr.
\end{equation*}

We wish to solve the equation $DF=0$ without any kind of conditions imposed.  This is done in the following proposition.

\begin{prop}\label{formal_kernel}
Let $D$ be the operator defined by (\ref{dirac_operator}) and acting in the Hilbert space $L^2\left(ST^2\setminus(\{0\}\times S^1)\right)\otimes \C^2$.   Then 
the kernel of $D$ consists of those  $F\in L^2(ST^2\setminus(\{0\}\times S^1))\otimes \C^2$ for which the coefficients of (\ref{FourierDec}) satisfy the following relations: for $m\neq0$ and any $n$
\begin{equation}\label{kernel_f}
f_{m,n+1}(r)  = \frac{m}{|m|}(-A_{m,n}I_{n+1}(|m|r) + B_{m,n}K_{n+1}(|m|r))
\end{equation}
and 
\begin{equation}\label{kernel_g}
g_{m,n}(r)  = A_{m,n}I_n(|m|r) + B_{m,n}K_n(|m|r),
\end{equation}
while if $m=0$ and any $n$ then
$f_{0,n}(r) = A_{0,n}r^{-n}$ and $g_{0,n}(r) = B_{0,n}r^n$.
Here $A_{m,n}, B_{m,n}$ are constants and $I_n, K_n$ are the modified Bessel functions of the first and second kind respectively.   
\end{prop}

\begin{proof}
To solve the equation $DF =0$, it is enough to solve the equation

\begin{equation*}
\sum_{m,n\in\Z}\left(
\begin{array}{cc}
m & e^{i\f}\left(\frac{\d}{\d r} - \frac{n}{r}\right) \\
e^{-i\f}\left(-\frac{\d}{\d r} + \frac{n}{r}\right) & -m
\end{array}\right) \left(
\begin{array}{c}
f_{m,n}(r) \\
g_{m,n}(r)
\end{array}\right) e^{in\f + im\theta} =0
\end{equation*}
by the Fourier decomposition (\ref{FourierDec}).   There are two cases to consider. First for $m\neq 0$ and any $n$, letting $t=|m|r$,  
$\tilde f_{m,n}(t)=f_{m,n}(t/|m|)$, and $\tilde g_{m,n}(t)=g_{m,n}(t/|m|)$ we arrive at the following system of differential equations:

\begin{equation*}
\left\{
\begin{aligned}
&\tilde f_{m,n+1}'(t) + \frac{n+1}{t}\tilde f_{m,n+1}(t) + \frac{m}{|m|}\tilde g_{m,n}(t) = 0 \\
&\tilde g_{m,n}'(t) - \frac{n}{t}\tilde g_{m,n}(t) +\frac{m}{|m|}\tilde f_{m,n+1}(t) = 0.
\end{aligned}\right. 
\end{equation*}
Substituting the second equation into the first yields:

\begin{equation*}
\tilde g_{m,n}''(t) + \frac{\tilde g_{m,n}'(t)}{t} - \left(1+ \frac{n^2}{t^2}\right)\tilde g_{m,n}(t) =0,
\end{equation*}
which is equation $(\ref{bessel_eq})$.   This implies that $\tilde g_{m,n}(t)$ is a linear combination of the two modified Bessel functions, i.e. $\tilde g_{m,n}(t) = A_{m,n}I_n(t) + B_{m,n}K_n(t)$.   Then, using the second equation from the above system and using the relations (\ref{bessel_rel_1}), we  get the desired result.   

If $m=0$ then the above system of differential equations reduces to:

\begin{equation*}
\left\{
\begin{aligned}
&f_{0,n}'(r) + \frac{n}{r}f_{0,n}(r) = 0 \\
&g_{0,n}'(r) - \frac{n}{r}g_{0,n}(r) = 0.
\end{aligned}\right. 
\end{equation*}
This system is uncoupled and the solutions are easily seen to be $f_{0,n}(r) = A_{0,n}r^{-n}$ and $g_{0,n}(r) = B_{0,n}r^n$.
Thus this completes the proof.
\end{proof}

In the following proposition we rephrase the domain condition of $D$ in terms of the Fourier coefficients to explicitly write down the boundary condition.

\begin{prop}\label{boundary_prop} 
Suppose that $F$ is in the domain of $D$ and has Fourier decomposition given by (\ref{FourierDec}). Then 
\begin{equation}\label{bound_equ} 
|m|K_{n+1}(|m|)g_{m,n}(1) - mK_n(|m|)f_{m,n+1}(1)=0
\end{equation}
if $m\ne 0$ and $f_{0,n}(1)=0$
for $n\leq 0$, and $g_{0,n}(1)=0$ for $n\geq 0$.
\end{prop}
\begin{proof}
We want to find a function $F^{ext}$ for $r\geq 1$ such that $DF^{ext} =0$ and $F^{ext}|_{r=1} = F|_{r=1}$.  First solving the system $DF^{ext} =0$ for $F^{ext}$ yields

\begin{equation*}
F^{ext} = \sum_{m,n\in\Z} \begin{pmatrix}f^{ext}_{m,n}(r) \\ g^{ext}_{m,n}(r)\end{pmatrix} e^{in\f+im\theta}
\end{equation*}
with $f^{ext}_{m,n+1}(r)$ and $g^{ext}_{m,n}(r)$  given by (\ref{kernel_f}) and (\ref{kernel_g}) for $m\ne 0$. Additionally $f^{ext}_{0,n}(r) = A_{0,n}r^{-n}$ and $g^{ext}_{0,n}(r) = B_{0,n}r^n$. First consider $m\ne 0$. In order for the solutions to agree on the boundary of the disk, the coefficients $A_{m,n}$ and $B_{m,n}$ must  solve the following system of equations

\begin{equation*}
\left(
\begin{array}{cc}
I_n(|m|) & K_n(|m|) \\
-\frac{m}{|m|}I_{n+1}(|m|) & \frac{m}{|m|}K_{n+1}(|m|)
\end{array}\right)\left(
\begin{array}{c}
A_{m,n} \\
B_{m,n}
\end{array}\right) = \left(
\begin{array}{c}
g_{m,n}(1)\\
f_{m,n+1}(1)
\end{array}\right).
\end{equation*}

The solution is: 
\begin{equation*}
A_{m,n} = |m|K_{n+1}(|m|)g_{m,n}(1) - mK_n(|m|)f_{m,n+1}(1)
\end{equation*} 
and
\begin{equation*} 
B_{m,n} = |m|I_{n+1}(|m|)g_{m,n}(1) + mI_n(|m|)f_{m,n+1}(1)
\end{equation*} 

If $m=0$ we get $f_{0,n}(1) = A_{0,n}$ and $g_{0,n}(1) = B_{0,n}$. 
 
Our boundary condition requires that $F^{ext}$ is square integrable on the complement of $ST^2$ in $\C\times S^1$. Because of the asymptotic properties of the modified Bessel functions, see (\ref{asym_infinity}), this forces $A_{m,n}=0$
for $m\ne 0$. If $m=0$ the integrability of powers of $r$ force $A_{0,n}=0$ for $n\leq 0$, and $B_{0,n}=0$ for $n\geq 0$. The statement follows from the above formulas for $A_{m,n}$, $A_{0,n}$, and $B_{0,n}$.
\end{proof}

Combining the above two propositions we obtain the following corollary.

\begin{cor}
Let $D$ be the operator defined by (\ref{dirac_operator}) subject to the boundary conditions (\ref{dirac_domain}).   Then its kernel is trivial.
\end{cor}
\begin{proof}
Let $F\in L^2(ST^2\setminus(\{0\}\times S^1))\otimes \C^2$ be a solution of $DF=0$, as in Proposition \ref{formal_kernel}. Its extension to a solution on $\C\times S^1\setminus(\{0\}\times S^1)$ is clearly given by the same formula.  Since the powers of $r$ are either not square integrable at zero or at infinity, it is clear that $A_{0,n} =0$ and $B_{0,n}=0$. 

If $m \neq 0$ then in order for $F$ to be regular at zero we must have $B_{m,n}= 0$ by the asymptotic expansion of $K_n(t)$ near zero, see (\ref{asym_zero}). Then Proposition \ref{boundary_prop} implies that $A_{m,n}=0$, hence the kernel of $D$ is trivial.   
\end{proof}

It turns out that our boundary condition is self-adjoint as demonstrated in the next proposition.

\begin{prop}
The operator $D$ defined by (\ref{dirac_operator}) subject to the boundary conditions (\ref{dirac_domain}) is self-adjoint.
\end{prop}
\begin{proof}
It is clear that $D$ is formally self-adjoint. From the standard elliptic theory \cite{BBW} the domain of $D$ and its adjoint consists of (vector-valued) functions of the first Sobolev class. Thus the only thing that we need to check are the boundary conditions of the adjoint. To this end we inspect the boundary integral in Green's formula.  Let $F,G$ be $H^1$ functions on the solid torus. Using Proposition 2.2 from \cite{KM1} and the Fourier decompositions:
\begin{equation*}
\begin{aligned}
&F = \left(
\begin{array}{c}
f \\ g
\end{array}\right) = \sum_{m,n\in\Z} \left(
\begin{array}{c}
f_{m,n}(r) \\ g_{m,n}(r)
\end{array}\right) e^{in\f + im\theta} \\
&G = \left(
\begin{array}{c}
p \\ q
\end{array}\right) = \sum_{m,n\in\Z} \left(
\begin{array}{c}
p_{m,n}(r) \\ q_{m,n}(r)
\end{array}\right) e^{in\f + im\theta},
\end{aligned}
\end{equation*}
we obtain:
\begin{equation*}
\begin{aligned}
\langle D G, F\rangle - \langle G,D F \rangle &= 2\left(\langle p, \d g/\d\overline{z} \rangle - \langle q, \d f/\d z \rangle - \langle \d q/\d\overline{z}, f\rangle + \langle \d p/\d z, g\rangle\right) \\
&= 2\sum_{m,n\in\Z}\left(\overline{p_{m,n+1}(1)}g_{m,n}(1) - \overline{q_{m,n}(1)}f_{m,n+1}(1)\right). \\
\end{aligned}
\end{equation*}
Now suppose that $F$ is in the domain of $D$, so it satisfies the conditions of Proposition \ref{boundary_prop}. For $G$ to be in the domain of the adjoint of $D$ we need the above expression to be equal to zero. This gives:
\begin{equation*}
\begin{aligned}
&\sum_{m\in\Z\setminus\{0\},n\in\Z}\left(\frac{m}{|m|}\frac{K_n(|m|)}{K_{n+1}(|m|)}\overline{p_{m,n+1}(1)} - \overline{q_{m,n}(1)}\right)f_{m,n+1}(1)+\\
&+\sum_{n< 0}\overline{p_{0,n+1}(1)}g_{0,n}(1) - \sum_{n\geq 0}\overline{q_{0,n}(1)}f_{0,n+1}(1)\ =0.
\end{aligned}
\end{equation*}
For $m\ne 0$ the above equation will equal zero for arbitrary $F$ only if $|m|K_{n+1}(|m|)q_{m,n}(1) - mK_n(|m|)p_{m,n+1}(1) =0$. Additionally we must have $p_{0,n}(1)=0$ for $n\leq 0$, and $q_{0,n}(1)=0$ for $n\geq 0$. All together those requirements are exactly the same as the original boundary condition. Hence $D$ and $D^*$ have the same domain and the proof is complete.
\end{proof}

The next goal is to construct the inverse of $D$.  This is done by explicit solving of a non-homogeneous system of differential equations for the Fourier components and adjusting the integration constants to get the regularity at $r=0$ and so that the boundary condition is satisfied.

\begin{prop}\label{Q_construction}
Let $D$ be the operator defined by (\ref{dirac_operator}) subject to the boundary conditions (\ref{dirac_domain}).   Then the operator $Q$ given by the formula (\ref{Q_formula}) below is the inverse to $D$, in other words $QD=DQ=I$.
\end{prop}

\begin{proof}
To compute the inverse of $D$ we solve the equation $DF=G$, which we will reduce to solving a non-homogeneous second order ordinary differential equation.   We will use the Fourier decomposition (\ref{FourierDec}). We first  consider the case $m\neq 0$.   Letting $t=|m|r$, $\tilde f_{m,n}(t)=f_{m,n}(t/|m|)$, and similarly for other functions, $DF=G$ becomes the system of differential equations:

\begin{equation}\label{non_homo_system}
\left\{
\begin{aligned}
&\tilde g_{m,n}'(t) - \frac{n}{t}\tilde g_{m,n}(t) + \frac{m}{|m|}\tilde f_{m,n+1}(t) = \frac{\tilde p_{m,n+1}(t)}{|m|} \\
&-\tilde f_{m,n+1}'(t) - \frac{n+1}{t}\tilde f_{m,n+1}(t) - \frac{m}{|m|}\tilde g_{m,n}(t) = \frac{\tilde q_{m,n}(t)}{|m|}.
\end{aligned}\right.
\end{equation}
Next by substituting the first equation into the second we arrive at the following second order differential equation:

\begin{equation}\label{non_homo_bessel_eq}
\begin{aligned}
&\tilde g_{m,n}''(t) + \frac{1}{t}\tilde g_{m,n}'(t) - \left(1+\frac{n^2}{t^2}\right)\tilde g_{m,n}(t) =\\
&= \frac{1}{m}\tilde q_{m,n}(t) + \frac{\tilde p_{m,n+1}'(t)}{|m|} + \frac{n+1}{t}\cdot \frac{\tilde p_{m,n+1}(t)}{|m|}=:h_{m,n}(t),
\end{aligned}
\end{equation}
where the right hand side of the above equation was denoted by $h_{m,n}(t)$.   Notice that equation $(\ref{non_homo_bessel_eq})$ is the non-homogeneous version of the Bessel differential equation $(\ref{bessel_eq})$.   General theory of ordinary differential equations tells us that the general solution to equation $(\ref{non_homo_bessel_eq})$ is

\begin{equation*}
\tilde g_{m,n}(t) = c_1^{(m,n)}(t)I_n(t) + c_2^{(m,n)}(t)K_n(t)
\end{equation*}
where  $c_1^{(m,n)}(t), c_2^{(m,n)}(t)$ solve the following system

\begin{equation*}
\begin{pmatrix}
I_n(t) & K_n(t) \\
I_n'(t) & K_n'(t)
\end{pmatrix}
\begin{pmatrix}
c_1^{(m,n)}(t) \\
c_2^{(m,n)}(t)
\end{pmatrix}' =
\begin{pmatrix}
0 \\
h_{m,n}(t)
\end{pmatrix}.
\end{equation*}
The solution of this system is 
\begin{equation*}
c_1^{(m,n)}(t) = A_{m,n}+ \int_{|m|}^t sK_n(s)h_{m,n}(s)ds\quad \textrm{ and }\quad c_2^{(m,n)}(t) =B_{m,n} -\int_0^t sI_n(s)h_{m,n}(s)ds,
\end{equation*}
where $A_{m,n}$ and $B_{m,n}$ are constants.
The boundary condition (\ref{bound_equ}) and the regularity at $t=0$ imply that we must have $c_1^{(m,n)}(t)$ equal to zero on the boundary, in other words where $t=|m|$.   The boundary condition and regularity also imply that $c_2^{(m,n)}(t)$ goes to zero as $t\to 0$. After performing the integration by parts we obtain that $A_{m,n}=K_n(|m|)p_{m,n+1}(|m|)$ and $B_{m,n}=0$, and \begin{equation*}
\begin{aligned}
c_1^{(m,n)}(t)&=\frac{1}{m}\int_{|m|}^t sK_n(s)\tilde q_{m,n}(s)ds + \frac{1}{|m|}\int_{|m|}^t sK_{n+1}(s)\tilde p_{m,n+1}(s)ds \\
c_2^{(m,n)}(t)&=-\frac{1}{m}\int_0^t sI_n(s)\tilde q_{m,n}(s)ds + \frac{1}{|m|}\int_0^t sI_{n+1}(s)\tilde p_{m,n+1}(s)ds.
\end{aligned}
\end{equation*}

Next we use the first equation of (\ref{non_homo_system}) to solve for the general $\tilde{f}_{m,n+1}(t)$ term to get:

\begin{equation*}
\frac{m}{|m|}\tilde f_{m,n+1}(t) = \frac{\tilde p_{m,n+1}(t)}{|m|} + \frac{n}{t}\tilde g_{m,n}(t) - \tilde g_{m,n}'(t). 
\end{equation*}
A straightforward calculation using the relations between the derivatives and indices of the modified Bessel functions (\ref{bessel_rel_1}) and the formula for the Wronskian of the modified Bessel functions (\ref{bessel_wronskian}) yields:

\begin{equation*}
\tilde f_{m,n+1}(t) = \frac{m}{|m|}\left(-c_1^{(m,n)}(t)I_{n+1}(t) + c_2^{(m,n)}(t)K_{n+1}(t)\right) .
\end{equation*}

Next consider the case when $m=0$.   The system of differential equations reduces to the following
\begin{equation*}
\left\{
\begin{aligned}
&g_{0,n}'(r) - \frac{n}{r}g_{0,n}(r) = p_{0,n+1}(r) \\
&f_{0,n+1}'(r) + \frac{n+1}{r}f_{0,n+1}(r) = -q_{0,n}(r)
\end{aligned}\right. 
\end{equation*}
This system is an uncoupled system and can be solved using an integration factor in each equation. 
Therefore the formula for the parametrix to $D$ is

\begin{equation}\label{Q_formula}
QG :=\sum_{m\in\Z\setminus\{0\},n\in\Z}\left(
\begin{array}{c}
f_{m,n}(r) \\ g_{m,n}(r)
\end{array}\right)
e^{in\f +im\theta} + \sum_{n\in\Z}\left(
\begin{array}{c}
f_{0,n}(r) \\ g_{0,n}(r)
\end{array}\right)
e^{in\f}
\end{equation}
where for $m\ne 0$:

\begin{equation*}
\begin{aligned}
&f_{m,n+1}(r) =\\
&= |m|I_{n+1}(|m|r)\int_{r}^{1}K_{n}(|m|\rho)q_{m,n}(\rho)\rho d\rho + mI_{n+1}(|m|r)\int_{r}^{1}K_{n+1}(|m|\rho)p_{m,n+1}(\rho)\rho d\rho \\
& - |m|K_{n+1}(|m|r)\int_0^{r}I_{n}(|m|\rho)q_{m,n}(\rho)\rho d\rho + mK_{n+1}(|m|r)\int_0^{r}I_{n+1}(|m|\rho)p_{m,n+1}(\rho)\rho d\rho 
\end{aligned}
\end{equation*}
\begin{equation*}
\begin{aligned}
&g_{m,n}(r) =\\
&= -mI_n(|m|r)\int_{r}^{1}K_n(|m|\rho)q_{m,n}(\rho)\rho d\rho - |m|I_n(|m|r)\int_{r}^{1}K_{n+1}(|m|\rho)p_{m,n+1}(\rho)\rho d\rho\\
&-mK_n(|m|r)\int_0^{r}I_n(|m|\rho)q_{m,n}(\rho)\rho d\rho + |m|K_n(|m|r)\int_0^{r}I_{n+1}(|m|\rho)p_{m,n+1}(\rho)\rho d\rho
\end{aligned}
\end{equation*}
and
\begin{equation*}
f_{0,n+1}(r) = \left\{
\begin{aligned}
-&\int_0^r \frac{\rho^{n}}{r^{n+1}} q_{0,n}(\rho)\rho d\rho \quad n\geq 0 \\
&\int_r^1 \frac{\rho^{n}}{r^{n+1}} q_{0,n}(\rho)\rho d\rho \quad n< 0,
\end{aligned}\right.
\end{equation*}
and
\begin{equation*}
g_{0,n}(r) = \left\{
\begin{aligned}
-&\int_r^1 \frac{r^{n}}{\rho^{n+1}} p_{0,n+1}(\rho)\rho d\rho \quad n\geq 0 \\
&\int_0^r \frac{r^{n}}{\rho^{n+1}} p_{0,n+1}(\rho)\rho d\rho \quad n< 0.
\end{aligned}\right.
\end{equation*}

It is is now a routine exercise to verify that $DQ=QD=I$. Thus this completes the proof.
\end{proof}

\section{The Parametrix}

Now that the parametrix has been constructed the next goal is to show that it is a compact operator.   This is the main result of the paper.   

\begin{theo}\label{cont_Q_theorem}
The Dirac operator $D$, defined by (\ref{dirac_operator}) and subject to the boundary conditions (\ref{dirac_domain}) has a bounded inverse.   Moreover that inverse is a compact operator.
\end{theo}
\begin{proof}

Consider the following integral operators in $L^2([0,1],rdr)$ for $i,j=0,1$ and $m\ne 0$:
\begin{equation*}
R_{ij}^{(m,n)}f(r):=|m|\int_{r}^{1}I_{n+i}(|m|r)K_{n+j}(|m|\rho)f(\rho)\rho d\rho,
\end{equation*}
\begin{equation*}
S_{ij}^{(m,n)}f(r):=|m|\int_{0}^{r}K_{n+i}(|m|r)I_{n+j}(|m|\rho)f(\rho)\rho d\rho,
\end{equation*}
and for $n\geq 0$:
\begin{equation*}
T_{1}^{(0,n)}f(r):=\int_r^1 \frac{r^n}{\rho^{n+1}} f(\rho)\,\rho d\rho,
\end{equation*}
\begin{equation*}
T_{2}^{(0,n)}f(r):=\int_0^r \frac{\rho^{n}}{r^{n+1}} f(\rho)\,\rho d\rho.
\end{equation*}

Then we can rewrite formula (\ref{Q_formula}) for $Q$ in the following way:

\begin{equation}\label{Q_direct_sum}
\begin{aligned}
f_{m,n+1} &= R_{10}^{(m,n)}q_{m,n} + \frac{|m|}{m}R_{11}^{(m,n)}p_{m,n+1} - S_{10}^{(m,n)}q_{m,n} + \frac{|m|}{m}S_{11}^{(m,n)}p_{m,n+1} \\
g_{m,n} &= -\frac{|m|}{m}R_{00}^{(m,n)}q_{m,n} - R_{01}^{(m,n)}p_{m,n+1} -\frac{|m|}{m}S_{00}^{(m,n)}q_{m,n} + S_{11}^{(m,n)}p_{m,n+1}
\end{aligned}
\end{equation}
and
\begin{equation*}
f_{0,n+1} = \left\{
\begin{aligned}
-T_{2}^{(0,n)}&q_{0,n} \quad n\geq 0 \\
T_{1}^{(0,-n-1)}&q_{0,n} \quad n< 0,
\end{aligned}\right.
\end{equation*}
and
\begin{equation*}
g_{0,n}= \left\{
\begin{aligned}
-T_{1}^{(0,n)}&p_{0,n+1} \quad n\geq 0 \\
T_{2}^{(0,-n-1)}&p_{0,n+1} \quad n< 0.
\end{aligned}\right.
\end{equation*}

We will show that all ten integral operators above are Hilbert-Schmidt by estimating the Hilbert-Schmidt norms. It turns out that the HS norm of each integral operator goes to zero as $|m|+|n|$ goes to infinity. This implies that $Q$ is a compact operator as the norm limit of compact operators, since it is (up to a shift in the $n$ index) a direct sum of compact operators with decreasing norms. 

To show that $T_{2}^{(0,n)}$ and $T_{1}^{(0,n)}$ are Hilbert-Schmidt we simply compute:
\begin{equation}\label{T_norm}
\|T_{1}^{(0,n)}\|_2^2 = \|T_{2}^{(0,n)}\|_2^2 = \int_0^1\int_0^r\left(\frac{\rho}{r}\right)^{2n+1}d\rho dr = \frac{1}{4(n+1)}.
\end{equation}
For the other operators we have:

\begin{equation*}
\|R_{ij}^{(m,n)}\|_2^2 = |m|^2\int_0^1\int_{r}^{1}I_{n+i}^2(|m|r)K_{n+j}^2(|m|\rho )\,r\rho d\rho dr,
\end{equation*}
and
\begin{equation*}
\|S_{ij}^{(m,n)}\|_2^2 = |m|^2\int_0^1\int_{0}^{r}K_{n+i}^2(|m|r)I_{n+j}^2(|m|\rho )\,r\rho d\rho dr.
\end{equation*}
Clearly we have $R_{11}^{(m,n)}=R_{00}^{(m,n+1)}$ and $S_{11}^{(m,n)}=S_{00}^{(m,n+1)}$ and additionally, using the inequality (\ref{bessel_n_inq}), we can conclude that: 

\begin{equation*}
\begin{aligned}
&\|R_{10}^{(m,n)}\|_2^2 \le \|R_{00}^{(m,n)}\|_2^2 \le \|R_{01}^{(m,n)}\|_2^2 \\
&\|S_{01}^{(m,n)}\|_2^2 \le \|S_{00}^{(m,n)}\|_2^2 \le \|S_{10}^{(m,n)}\|_2^2.
\end{aligned}
\end{equation*}
Consequently, we only have to estimate the Hilbert-Schmidt norm for $R_{01}^{(m,n)}$ and $S_{10}^{(m,n)}$.    We have:

\begin{equation*}
\|R_{01}^{(m,n)}\|_2^2 =  \frac{1}{|m|^2}\int_0^{|m|}\int_t^{|m|}K_{n+1}^2(s)I_n^2(t)stdsdt=\frac{1}{|m|^2}\int_0^{|m|}\int_0^s K_{n+1}^2(s)I_n^2(t)stdtds,
\end{equation*}
where we've changed to new variables $t=|m|r$, $s=|m|\rho$,  and used Fubini's Theorem.  We will now estimate the above expression in two ways to show that it goes to zero when $|m|+|n|$ increases.
  
From (\ref{log_deriv_0}) we have $I_n^2(t) \leq \frac{t}{n}I_n(t)I_n'(t)$ which by integration yields:
\begin{equation}\label{bessel_sqr_1}
\int_0^sI_n^2(t)dt \leq  \frac{s}{2n}I_n^2(s).
\end{equation}
Next,  using $t\leq s$, (\ref{bessel_sqr_1}), and (\ref{bessel_product_inq}), we get for $n\ne 0$:
\begin{equation*}
\|R_{01}^{(m,n)}\|_2^2\leq \frac{1}{2|n||m|^2}\int_0^{|m|} s^3K_{n+1}^2(s)I_n^2(s)ds\le \frac{1}{2|n||m|^2}\int_0^{|m|} sds = \frac{1}{4|n|}.
\end{equation*}
On the other hand from (\ref{log_deriv_1}) we have $I_n^2(t)\leq {2}I_n(t)I_n'(t)$, yielding:
\begin{equation}\label{bessel_sqr_2}
\int_0^sI_n^2(t)dt\leq I_n^2(s) .
\end{equation}
So using inequalities (\ref{bessel_sqr_2}) and (\ref{bessel_product_inq}) again, we get for $n\ne 0$:
\begin{equation*}
\|R_{01}^{(m,n)}\|_2^2\leq \frac{1}{|m|^2}\int_0^{|m|} s^2K_{n+1}^2(s)I_n^2(s)ds\le \frac{1}{|m|^2}\int_0^{|m|} ds = \frac{1}{|m|}.
\end{equation*}
Finally, if $n=0$, we notice that the recurrence relations (\ref{bessel_rel_1}) imply:
\begin{equation*}
tI_0^2(t)=\left(tI_1(t)I_0(t)\right)'-tI_1^2(t)\leq \left(tI_1(t)I_0(t)\right)'.
\end{equation*}
Hence we get an integral estimate: 
\begin{equation*}
\int_0^s I_0^2(t)tdt\leq sI_1(s)I_0(s)\leq sI_0^2(s),
\end{equation*}
which we use to estimate the norm above as follows:
\begin{equation*}
\|R_{01}^{(m,0)}\|_2^2\leq \frac{1}{|m|^2}\int_0^{|m|} s^2K_{1}^2(s)I_0^2(s)ds\leq  \frac{1}{|m|}.
\end{equation*}

For the norm of $S_{10}^{(m,n)}$ we observe that, after a change of variables, we have:

\begin{equation*}
\|S_{10}^{(m,n)}\|_2^2 =  \frac{1}{|m|^2}\int_0^{|m|}\int_0^tI_n^2(s)K^2_{n+1}(t)stdsdt = \|R_{01}^{(m,n)}\|_2^2.
\end{equation*}
This shows that all of the operators are indeed Hilbert-Schmidt operators. Moreover we have the estimates:
\begin{equation}\label{R_norm}
\|R_{ij}^{(m,n)}\|_2^2\leq\frac{\textrm{const}}{\sqrt{1+m^2+n^2}},
\end{equation}
and
\begin{equation}\label{S_norm}
\|S_{ij}^{(m,n)}\|_2^2\leq\frac{\textrm{const}}{\sqrt{1+m^2+n^2}}.
\end{equation}
It follows by the remarks at the beginning of the proof that $Q$ is compact. Thus the proof of the theorem is complete.

\end{proof}

\begin{theo}\label{Schatten_Q_theorem}
The  operator $Q$, defined by (\ref{Q_formula}), is a p-th Schatten-class operator for all $p>3$. 
\end{theo}
\begin{proof}
Notice that the $p$-th Schatten norm of $Q$ can be estimated as follows:
\begin{equation}\label{Qp_estimate}
\|Q\|_p^p\leq \textrm{const}\sum_{m,n,i,j}\left(\|R_{ij}^{(m,n)}\|_p^p+\|S_{ij}^{(m,n)}\|_p^p\right)+
 \sum_{n,i}\|T_{i}^{(0,n)}\|_p^p.
\end{equation}
This is because $Q$ is (essentially) a direct sum of two by two matrices with entries made up of the ten integral operators we studied above, see (\ref{Q_direct_sum}).

To bound $\|R_{ij}^{(m,n)}\|_p^p$ and the other norms we use the following  interpolation estimate for the  $p$-th Schatten norm: if $a$ is a Hilbert-Schmidt operator and $p\geq 2$ then
\begin{equation}\label{log_convex}
\|a\|_p^p\leq \|a\|_2^2\,\|a\|^{p-2}.
\end{equation}
The estimate easily follows from the definition  of the p-th Schatten norm. We have already obtained estimates on the Hilbert-Schmidt norms of $R_{ij}^{(m,n)}$ and the other operators in (\ref{R_norm}), (\ref{S_norm}), and (\ref{T_norm}), so by the above interpolation we need estimates on the operator norms.
The main tool used to establish such estimates for the operator norms of integral operators is the Schur-Young inequality, see \cite{HS}.   It is stated in the lemma below.

\begin{lem}(Schur-Young Inequality)\label{shuryounginq}
Let $\mathcal{K}: L^2(Y) \longrightarrow L^2(X)$ be an integral operator:

\begin{equation*}
\mathcal{K}f(x) = \int K(x,y)f(y)dy
\end{equation*}
Then we have:

\begin{equation*}
\|\mathcal{K}\|^2 \le \left(\underset{x\in X}{\textrm{sup}}\int_Y |K(x,y)|dy\right)\left(\underset{y\in Y}{\textrm{sup}}\int_X |K(x,y)|dx\right).
\end{equation*}
\end{lem}

The kernels of our integral operators are products of modified Bessel functions, and the difficulty here is to estimate the integrals of such products. The main technical step in those estimates is summarized in the following lemma.

\begin{lem}\label{int_bndedness}
Consider the following  expressions for $m\neq0$: 
\begin{equation*}
\mathcal{I}_1^{(m,n)} = \frac{1}{|m|}\underset{0\leq s\leq |m|}{\textrm{ sup }}\int_0^s K_{n+1}(s)I_n(t)tdt,
\end{equation*}
\begin{equation*}
\mathcal{I}_2^{(m,n)} = \frac{1}{|m|}\underset{0\leq t\leq |m|}{\textrm{ sup }}\int_t^{|m|} K_{n+1}(s)I_n(t)sds.
\end{equation*}
There is a constant such that for $i=1,2$:
\begin{equation*}
\mathcal{I}_i^{(m,n)}\leq\frac{\textrm{const}}{\sqrt{1+m^2+n^2}}.
\end{equation*}
\end{lem}
The proof of this lemma will be postponed until we finish the main line of the argument.
Now we turn to estimating $\|R_{ij}^{(m,n)}\|$ for $m\ne 0$. Using Lemma \ref{shuryounginq} we have:

\begin{equation*}
\|R_{ij}^{(m,n)}\|^2 \le m^2\left(\isup{0}{r}{1}\int_{r}^{1}K_{n+i}(|m|\rho)I_{n+j}(|m|r)\rho d\rho\right)\left(\isup{0}{\rho}{1}\int_0^{\rho}K_{n+i}(|m|\rho)I_{n+j}(|m|r)rdr\right)
\end{equation*}
Changing variables in both integrals we get:

\begin{equation*}
\|R_{ij}^{(m,n)}\|^2 \le \left(\frac{1}{|m|}\isup{0}{t}{|m|}\int_{t}^{|m|}K_{n+i}(s)I_{n+j}(t)sds\right)\left(\frac{1}{|m|}\isup{0}{s}{|m|}\int_0^s K_{n+i}(s)I_{n+j}(t)tdt\right).
\end{equation*}
By the monotonicity (\ref{bessel_n_inq}) the right hand side is biggest when $i=1$ and $j=0$, and so
\begin{equation*}
\|R_{ij}^{(m,n)}\|^2\le\mathcal{I}_1^{(m,n)}\cdot\mathcal{I}_2^{(m,n)}.
\end{equation*}
It follows from Lemma \ref{int_bndedness} that:
\begin{equation}\label{R_op_norm}
\|R_{ij}^{(m,n)}\|\leq\frac{\textrm{const}}{\sqrt{1+m^2+n^2}}.
\end{equation}

We turn our attention to estimating $\|S_{ij}^{(m,n)}\|$.   By Lemma \ref{shuryounginq} we have:
\begin{equation*}
\|S_{ij}^{(m,n)}\|^2 \le m^2\left(\isup{0}{r}{1}\int_0^{r}I_{n+i}(|m|\rho)K_{n+j}(|m|r)\rho d\rho\right)\left(\isup{0}{\rho}{1}\int_{\rho}^1 I_{n+i}(|m|\rho)K_{n+j}(|m|r)rdr\right)
\end{equation*}
Clearly the expression on the right hand side of the above inequality is the same as the expression in the estimate of $\|R_{ij}^{(m,n)}\|^2$.
It follows that 
\begin{equation*}
\|S_{ij}^{(m,n)}\|\leq\frac{\textrm{const}}{\sqrt{1+m^2+n^2}}. 
\end{equation*}

Consider now the case $m=0$.   
We use Lemma \ref{shuryounginq} once again to compute:

\begin{equation*}
\|T_{2}^{(0,n)}\|^2 \le \left(\isup{0}{r}{1}\int_0^r \left(\frac{\rho}{r}\right)^{n+1}d\rho\right)\left(\isup{0}{\rho}{1}\int_\rho^1 \left(\frac{\rho}{r}\right)^{n}dr\right) \le \frac{\textrm{const}}{1+n^2}
\end{equation*}
and similarly

\begin{equation*}
\|T_{1}^{(0,n)}\|^2 \le \left(\isup{0}{r}{1}\int_r^1 \left(\frac{r}{\rho}\right)^{n}d\rho\right)\left(\isup{0}{\rho}{1}\int_0^\rho \left(\frac{r}{\rho}\right)^{n+1}dr\right) \le \frac{\textrm{const}}{1+n^2}.
\end{equation*}
Either way we have for $i=1,2$:

\begin{equation*}
\|T_i^{(0,n)}\| \le \frac{\textrm{const}}{\sqrt{1+n^2}}.
\end{equation*}
Combining (\ref{R_norm}), (\ref{R_op_norm}), and using (\ref{log_convex}) we get:
\begin{equation*}
\|R_{ij}^{(m,n)}\|_p^p\leq\frac{\textrm{const}}{(1+m^2+n^2)^{\frac{p-1}{2}}}\ ,
\end{equation*}
and exactly the same estimates for $\|S_{ij}^{(m,n)}\|_p^p$ and $\|T_i^{(0,n)}\|_p^p$. Consequently, by (\ref{Qp_estimate}) we get:
\begin{equation*}
\|Q\|_p^p\leq \sum_{m,n}\frac{\textrm{const}}{(1+m^2+n^2)^{\frac{p-1}{2}}}\ ,
\end{equation*}
where the series is summable when $p>3$. This concludes the proof of the theorem.
\end{proof}

\begin{proof} (of Lemma \ref{int_bndedness})

Since both $K_n(z)$ and $I_n(z)$ are symmetric for positive and negative $n$, see (\ref{bessel_sym}), we will only need to consider the case when $n\geq 0$.   

Using (\ref{log_deriv_1}) and integrating by parts we have:
\begin{equation*}
\int_0^sI_n(t)tdt\leq 2\int_0^sI'_{n+1}(t)tdt =2sI_{n+1}(s)-2\int_0^sI_{n+1}(t)dt\leq 2sI_{n+1}(s).
\end{equation*}
Consequently we get:
\begin{equation*}
\mathcal{I}_1^{(m,n)}=\frac{1}{|m|}\underset{0\leq s\leq |m|}{\textrm{ sup }}\int_0^s K_{n+1}(s)I_n(t)tdt \leq  \frac{1}{|m|}\underset{0\leq s\leq|m|}{\textrm{ sup }}2sK_{n+1}(s)I_{n+1}(s).
\end{equation*}
Now we bound $\mathcal{I}_1^{(m,n)}$ in two different ways. First we observe:
\begin{equation*}
\mathcal{I}_1^{(m,n)}\leq 2 \underset{0\leq s\leq|m|}{\textrm{ sup }}K_{n+1}(s)I_{n+1}(s)\leq\frac{2}{n+1},
\end{equation*}
by (\ref{bessel_product_zero}). On the other hand we have:
\begin{equation*}
\mathcal{I}_1^{(m,n)}\leq  \frac{2}{|m|}\underset{0\leq s\leq\infty}{\textrm{ sup }}sK_{n+1}(s)I_{n+1}(s)\leq \frac{2}{|m|},
\end{equation*}
by inequality (\ref{bessel_product_inq}).
It follows that $\mathcal{I}_1^{(m,n)} \leq \textrm{const}/\sqrt{1+m^2+n^2}$.

We estimate $\mathcal{I}_2^{(m,n)}$ in the same fashion, however the process  
is somewhat more complicated. Using  (\ref{log_deriv_2}) and integrating by parts we get:
\begin{equation*}
\int_t^{|m|} sK_{n+1}(s)ds \leq -2\int_t^{|m|}sK_n'(s)ds\leq 
2tK_n(t) + \int_t^{|m|} K_{n}(s)ds.
\end{equation*}
Using  (\ref{log_deriv_2}) again yields:
\begin{equation*}
\int_t^{|m|} sK_{n+1}(s)ds \leq 2tK_n(t) + 4K_{n}(t).
\end{equation*}
It follows that:

\begin{equation*}
\mathcal{I}_2^{(m,n)} = \frac{1}{|m|}\underset{0\leq t\leq |m|}{\textrm{ sup }}\int_t^{|m|} K_{n+1}(s)I_n(t)sds \leq 
 \frac{1}{|m|}\underset{0\leq t\leq |m|}{\textrm{ sup }}\left(2tK_n(t)I_n(t) + 4K_{n}(t)I_n(t)\right).
\end{equation*}
If $n>0$ we estimate the above expression in two ways using (\ref{bessel_product_inq})  and (\ref{bessel_product_zero}).
First we have:
\begin{equation*}
\mathcal{I}_2^{(m,n)}\leq \frac{1}{|m|}(2|m|+4)\frac{1}{2n}.
\end{equation*}
Secondly:
\begin{equation*}
\mathcal{I}_2^{(m,n)}\leq  \frac{1}{|m|}\left(2+\frac{4}{2n}\right).
\end{equation*}
If $n=0$ we have:
\begin{equation*}
\mathcal{I}_2^{(m,n)} \leq \frac{1}{|m|}\underset{0\leq t<\infty}{\textrm{ sup }}I_0(t)\int_t^{\infty} K_{1}(s)sds,
\end{equation*}
and we need to show that the function $I_0(t)\int_t^{\infty} K_{1}(s)sds$ is bounded. It follows from the asymptotic behavior
(\ref{asym_zero}) and (\ref{asym_infinity}) that the limit of $I_0(t)\int_t^{\infty} K_{1}(s)sds$ at $t=0$ is $\int_0^{\infty} K_{1}(s)sds<\infty$. On the other hand using L'Hospital's rule we get:
\begin{equation*}
\lim_{t\to\infty}I_0(t)\int_t^{\infty} K_{1}(s)sds=\lim_{t\to\infty}\frac{I_0^2(t)K_1(t)t}{I_0'(t)}=\lim_{t\to\infty}\frac{I_0^2(t)K_1(t)t}{I_1(t)}=\frac{1}{2},
\end{equation*}
by (\ref{asym_infinity}) again.
Thus, in a similar fashion to $\mathcal{I}_1^{(m,n)}$, we have  $\mathcal{I}_2^{(m,n)} \leq \textrm{const}/\sqrt{1+m^2+n^2}$. 
Therefore the proof of the lemma is complete.
\end{proof}
In conclusion we would like to remark that a somewhat more complicated proof of Lemma \ref{int_bndedness} is possible without the use of the  non-elementary inequality (\ref{product_mon}). Estimating along the lines of the Hilbert-Schmidt norm bound in the proof of Theorem \ref{cont_Q_theorem} it is enough to employ the inequalities (\ref{log_deriv_0}),
(\ref{log_deriv_1}), (\ref{log_deriv_2}), and (\ref{log_deriv_3}) instead of monotonicity of $K_n(t)I_n(t)$.

\appendix
\section{}
This section contains all of the relevant information on the modified Bessel functions.   The references that are used are \cite{AS}, \cite{Baricz}, \cite{MP}, and \cite{Soni}.   A short argument will be given for those results that are not from any of these references.   

\subsection{Basic properties}
The main reference of this subsection is \cite{AS}. 

The modified Bessel functions of integer order $n$ can be defined by the following expressions:

\begin{equation*}
I_n(t) = \frac{1}{\pi}\int_0^\pi e^{t\cos\alpha}\cos(n\alpha)\ d\alpha
\end{equation*}
and

\begin{equation}\label{K_def}
K_n(t) = \int_0^\infty e^{-t\cosh\alpha}\cosh(n\alpha)\ d\alpha
\end{equation}
where in both formulas $t$ is a positive real number.  

Both  functions are symmetric in $n$:

\begin{equation}\label{bessel_sym}
I_n(t) = I_{-n}(t) \textrm{ and }K_n(t) = K_{-n}(t).
\end{equation}
Consequently, without the loss of generality, it will be assumed that $n$ is a non-negative integer.

We have the following power series representation for $I_n(t)$:
\begin{equation*}
I_n(t)=\sum_{k=0}^\infty\frac{\left(t/2\right)^{n+2k}}{k!(n+k)!}
\end{equation*}
It follows that both modified Bessel functions are positive.

They are two independent solutions of the second-order differential equation:

\begin{equation}\label{bessel_eq}
\frac{d^2x}{dt^2} + \frac{1}{t}\frac{dx}{dt} - \left(1+ \frac{n^2}{t^2}\right)x(t) =0
\end{equation}
which is called the modified Bessel equation.      

They satisfy the recurrence relations with derivatives:

\begin{equation}\label{bessel_rel_1}
I_n'(t) = I_{n+1}(t) + \frac{n}{t}I_n(t) \textrm{ and } K_n'(t) = -K_{n+1}(t) + \frac{n}{t}K_n(t),
\end{equation}
as well as:
\begin{equation}\label{bessel_rel_2}
I_n'(t) = I_{n-1}(t) - \frac{n}{t}I_n(t) \textrm{ and } K_n'(t) = -K_{n-1}(t) - \frac{n}{t}K_n(t),
\end{equation}

The Wronskian of the two functions is:

\begin{equation}\label{bessel_wronskian}
W(K_n(t),I_n(t)) = \det \left(
\begin{array}{cc}
K_n(t) & I_n(t) \\
K_n'(t) & I_n'(t)
\end{array}\right) = I_n(t)K_{n+1}(t) + I_{n+1}(t)K_n(t) = 1/t.
\end{equation}

They have the following expansions near zero for $n\geq 0$:

\begin{equation}\label{asym_zero}
I_n(t) \sim \frac{1}{\Gamma(n+1)}\left(\frac{t}{2}\right)^n\textrm{ and } K_n(t) \sim \left\{
\begin{array}{cc}
-\ln\left(\frac{t}{2}\right) -\gamma &\textrm{ if } n=0 \\
\frac{\Gamma(n)}{2}\left(\frac{2}{t}\right)^n &\textrm{ if } n>0
\end{array}\right.
\end{equation}
where $\gamma$ is the Euler-Mascheroni constant.  The expansions at infinity are:

\begin{equation}\label{asym_infinity}
I_n(t) \sim \frac{e^t}{\sqrt{2\pi t}} \textrm{ and } K_n(t) \sim e^{-t}\sqrt{\frac{\pi}{2t}}.
\end{equation}

In the following subsections we collect less known results about the modified Bessel functions. 

\subsection{Monotonicity}

The modified Bessel functions have simple monotonicity properties in the argument $t$: $I_n'(t)>0$ and $K_n'(t)\leq 0$ on $(0,\infty)$, which says that $I_n(t)$ is increasing and $K_n(t)$ is decreasing.  The first inequality follows from (\ref{bessel_rel_1}). The second inequality follows immediately from the integral representation (\ref{K_def}).   

Additionally there are the following monotonicity properties in the order $n$:

\begin{equation}\label{bessel_n_inq}
I_{n+1}(t) \le I_n(t) \textrm{ and } K_n(t) \le K_{n+1}(t).
\end{equation}
The first inequality was proven in \cite{Soni}.  It also follows from Turan - type inequality \cite{Baricz}:
\begin{equation*}
I_{n-1}(t)I_{n+1}(t)-I_{n}^2(t)\leq 0.
\end{equation*}
For the second inequality we estimate
\begin{equation*}
\begin{aligned}
&K_{n+1}(t) = \int_0^\infty e^{-t\cosh\alpha}\left(\cosh(n\alpha)\cosh\alpha + \sinh(n\alpha)\sinh\alpha\right)d\alpha\ge \\
&\ge \int_0^\infty e^{-t\cosh\alpha} \cosh(n\alpha)\cosh\alpha d\alpha\ge \int_0^\infty e^{-t\cosh\alpha} \cosh(n\alpha) d\alpha = K_n(t).
\end{aligned}
\end{equation*}

We also have monotonicity of the product: 
\begin{equation}\label{product_mon}
\left(K_n(t)I_n(t)\right)'\leq 0
\end{equation}
i.e. $K_n(t)I_n(t)$ is a decreasing function of $t$, see \cite{MP}.

\subsection{Product estimates}

For $n\geq 1$ we have:

\begin{equation*}
\lim_{t\to0^+} K_n(t)I_n(t) = \frac{1}{2n}.
\end{equation*}
This is a simple consequence of the asymptotics of $I_n(t)$ and $K_n(t)$ as $t\to 0$, see (\ref{asym_zero}).
Since $I_n(t)K_n(t)$ is decreasing on $(0,\infty)$, we have:

\begin{equation}\label{bessel_product_zero}
K_n(t)I_n(t) \leq  \frac{1}{2n}.
\end{equation}
Additionally we have:

\begin{equation}\label{bessel_product_inq}
tK_{n}(t)I_n(t)\leq tK_{n+1}(t)I_n(t) \leq 1.
\end{equation}
The inequality follows from (\ref{bessel_n_inq}) and from the Wronskian formula (\ref{bessel_wronskian}) since both terms on the left-hand side of that equation are positive.

\subsection{Derivative estimates}

The proofs in this paper use the following two inequalities with derivatives of the modified Bessel functions of the first kind. They are, for $n>0$:
\begin{equation}\label{log_deriv_0}
I_n(t)\leq  \frac{t}{n}I_n'(t),
\end{equation}
and
\begin{equation}\label{log_deriv_1}
I_{n-1}\leq I_n(t)\leq  2\,I_n'(t).
\end{equation}
To prove them we notice that from (\ref{bessel_rel_1}) we get $I_n'(t) - \frac{n}{t}I_n(t) >0$, which gives (\ref{log_deriv_0}). Secondly,
(\ref{bessel_rel_1}) and (\ref{bessel_rel_2}) give:
\begin{equation*}
2\,I_n'(t)=I_{n+1}(t)+I_{n-1}(t)\geq I_{n-1}(t)\geq I_{n}(t),
\end{equation*}
which is (\ref{log_deriv_1}).

Both inequalities above are also a direct consequence of the following stronger result of \cite{MP}:
\begin{equation*}
\frac{tI_n'(t)}{I_n(t)} > \sqrt{t^2\frac{n}{n+1}+n^2}
\end{equation*}

For the modified Bessel functions of the second kind we have the following useful result from \cite{MP}:

\begin{equation*}
\frac{sK_n'(s)}{K_n(s)}\leq  -\sqrt{s^2+n^2}. 
\end{equation*}
An analog of (\ref{log_deriv_0}) and obtainable in the same way from (\ref{bessel_rel_2}) is:
\begin{equation}\label{log_deriv_3}
K_n(s)\leq  -\frac{s}{n}K_n'(s).
\end{equation}
However for our applications we only need the following estimate:
combining (\ref{bessel_rel_1}) and (\ref{bessel_rel_2}) gives:
\begin{equation}\label{log_deriv_2}
-2\,K_n'(t)=K_{n+1}(t)+K_{n-1}(t)\geq K_{n+1}(t)\geq K_{n}(t).
\end{equation}

\end{document}